\documentclass[a4paper,10pt,reqno]{amsart}
\usepackage[utf8]{inputenc}
\usepackage{amsaddr}
\usepackage{calrsfs}
\usepackage{graphicx,color}

\numberwithin{equation}{section}

\theoremstyle{definition}
\newtheorem{definition}{Definition}[section]
  
\theoremstyle{plain}
\newtheorem{theorem}{Theorem}[section]
\newtheorem{proposition}{Proposition}[section]
\newtheorem{corollary}{Corollary}[section]
\newtheorem{lemma}{Lemma}[section]
\theoremstyle{definition}
\newtheorem{remark}{Remark}[section]
\newtheorem{example}{Example}[section]

\newcommand{\E}{\mathbb{E}}
\renewcommand{\d}{\mathrm{d}}
\newcommand{\R}{\mathbb{R}}
\newcommand{\1}{\mathrm{1}}
\newcommand{\K}{\mathrm{K}}
\newcommand{\F}{\mathcal{F}}
\newcommand{\Cov}{\mathbb{C}\mathrm{ov}}
\newcommand{\Var}{\mathbb{V}\mathrm{ar}}
\newcommand{\p}{\partial}
\newcommand{\e}{\mathrm{e}}
\newcommand{\Q}{\mathbb{Q}}

\begin{document}

\title[Conditional-Mean Hedging Under Transaction Costs]{Conditional-Mean Hedging Under Transaction Costs in Gaussian Models}

\author[Sottinen]{Tommi Sottinen}
\address{Department of Mathematics and Statistics, University of Vaasa, P.O. Box 700, FIN-65101 Vaasa, FINLAND}
\email{tommi.sottinen@iki.fi}

\author[Viitasaari]{Lauri Viitasaari}
\address{Department of Mathematics and System Analysis, Aalto University School of Science, Helsinki, P.O. Box 11100, FIN-00076 Aalto,  FINLAND} 
\email{lauri.viitasaari@aalto.fi}

\begin{abstract}
We consider so-called regular invertible Gaussian Volterra processes and derive a formula for their prediction laws. Examples of such processes include the fractional Brownian motions and the mixed fractional Brownian motions.  As an application, we consider conditional-mean hedging under transaction costs in Black--Scholes type pricing models where the Brownian motion is replaced with a more general regular invertible Gaussian Volterra process.
\end{abstract}


\maketitle

\section{Introduction}

We consider discrete imperfect hedging under proportional transaction costs in Black--Scholes type pricing models where the asset price is driven by a relatively general Gaussian process; a so-called regular invertible Gaussian Volterra process.  These are continuous Gaussian processes that are non-anticipative linear transformations of continuous Gaussian martingales. 

For European vanilla type options we construct the so-called conditional-mean hedge.  This means that at each trading time the value of the conditional mean of the discrete hedging strategy coincides with the frictionless price.  By frictionless we mean the continuous-trading hedging price without transaction costs. The key ingredient in constructing the conditional mean hedging strategy is a representation for the regular conditional laws of regular invertible Gaussian Volterra processes which we provide in Section \ref{sect:prediction}.  Let us note that in our models there may be arbitrage strategies with continuous trading without transaction costs, but not with discrete trading strategies, even in the absence of trading costs.  

For the classical Black--Scholes model driven by the Brownian motion, the study of hedging under transaction costs goes back to Leland \cite{Leland-1985}. See also
Denis and Kabanov \cite{Denis-Kabanov-2010} and Kabanov and Safarian \cite{Kabanov-Safarian-2009} for a mathematically rigorous treatment.
For the fractional Black--Scholes model driven by the long-range dependent fractional Brownian motion, the study of hedging under transaction costs was studied in Azmoodeh \cite{Azmoodeh-2013}.
In the series of articles \cite{Shokrollahi-Kilicman-Magdziarz-2016,Wang-2010a,Wang-2010b,Wang-Yan-Tang-Zhu-2010,Wang-Zhu-Tang-Yan-2010} the discrete hedging in the fractional Black--Scholes model was studied by using the economically dubious Wick--It\^o--Skorohod interpretation of the self-financing condition.  Actually, with the economically solid forward-type pathwise interpretation of the self-financing condition, these hedging strategies are valid, not for the geometric fractional Brownian motion, but for a geometric Gaussian process where the driving noise is a Gaussian martingale with the same variance function as the corresponding fractional Brownian motion would have, see \cite{Gapeev-Sottinen-Valkeila-2011}.
Our approach here builds on the works \cite{Sottinen-Viitasaari-2017b} and \cite{Shokrollahi-Sottinen-2017}.  The novelty of this note is twofold: First, we extend the results to a more general class of Gaussian processes than just the long-range dependent fractional Brownian motions.  Second, we emphasize the models where there exists a non-trivial quadratic variation.  This makes the formulas and the analysis very different from the long-range dependent fractional Brownian case.

The rest of the paper is organized as follows:
In Section \ref{sect:pricing-model} we introduce our pricing model with a regular invertible Gaussian Volterra process as the driving noise, and develop a transfer principle for the noise.
In Section \ref{sect:arbitrage-completeness-qv} we investigate arbitrage and completeness in our pricing models.
In Section \ref{sect:prediction} we provide prediction formulas for the driving noise and for Markovian functionals of the asset price.
Finally, in Section \ref{sect:cmh} we provide formulas for conditional-mean hedging under transaction costs.

\section{Pricing Model with Invertible Gaussian Volterra Noise}\label{sect:pricing-model}

Let $T>0$ be a fixed time of maturity of the contingent claim under consideration.
We are interested in imperfect hedging in a geometric Gaussian model where the discounted risky asset follows the dynamics 
\begin{equation}\label{eq:risky-dynamics}
\frac{\d S_t}{S_t} 
= \d\mu(t) + \d X_t, \quad t\in[0,T],
\end{equation}
where $\mu\colon[0,T]\to\R$ is a known excess return of the asset and $X$ is a driving Gaussian noise.  We assume that $\mu$ is continuous with bounded variation.
For the noise $X$ we assume that it is continuous and centered with $X_0=0$ and covariance function
\begin{equation}\label{eq:cov}
R(t,s) = \E\left[X_tX_s\right], \quad s,t\in[0,T].
\end{equation}

To analyze the pricing model \eqref{eq:risky-dynamics}, we make the following rather technical Definition  \ref{dfn:rigv} that ensures the invertible Volterra representation and continuous quadratic variation for the noise process $X$.  We note that Definition \ref{dfn:rigv} is not very restrictive: many interesting Gaussian models satisfy it (see Example \ref{ex} below).  

\begin{definition}[Regular Invertible Gaussian Volterra Process]\label{dfn:rigv}
A centered Gaussian process over an interval $[0,T]$ with covariance function $R$ is a \emph{regular invertible Gaussian Volterra process} if
\mbox{}
\begin{enumerate}
\item
There exists a continuous increasing function $m\colon[0,T]\to\R_+$ and a Volterra kernel $K\in L^2([0,T]^2, \d m\times \d m)$ non-decreasing in the first variable that is partially continuously differentiable outside the diagonal and continuously differentiable on the diagonal, such that
\begin{equation*}\label{eq:cholesky}
R(t,s) = \int_0^{t\wedge s} K(t,u)K(s,u)\, \d m(u).
\end{equation*}
\item
Define
$$
\K^*[f](s) = f(s)K(T,s) + \int_s^T \left[ f(t) - f(s)\right] \frac{\p K}{\p t}(t, s)\, \d t.
$$
Then, for each $t\in[0,T]$, the equation
\begin{equation*}\label{eq:cholesky-inverse}
\K^*[f](s) = \1_{[0,t)}(s)
\end{equation*}
has a solution.
\end{enumerate}
\end{definition}

\begin{example}[Examples and Counterexamples]\label{ex}
\mbox{}
\begin{enumerate}
\item
Obviously, any continuous Gaussian martingale is a regular invertible Gaussian Volterra process.
\item
Fractional Brownian motions with Hurst index $H\in[1/2,1)$ are regular invertible Gaussian Volterra processes.  See e.g. Mishura \cite{Mishura-2008}, Section 1.8, for details.
\item
Mixed fractional Brownian motions with Hurst index $H\in[1/2,1)$ are regular invertible Gaussian Volterra processes.  See Cai, Chigansky and Kleptsyna \cite{Cai-Chigansky-Kleptsyna-2016} for details. 
\item
Fractional or mixed fractional Brownian motions with Hurst index $H\in(0,1/2)$ are not regular invertible Gaussian Volterra processes, since they have infinite quadratic variation, cf. Lemma \ref{lmm:qv} below. 
\item
The Gaussian slope $X_t = t\xi$, where $\xi$ is a standard Gaussian random variable is an invertible Gaussian Volterra process in the sense that it is generated non-anticipatively from a Gaussian martingale.  It is not regular, however, since the generating martingale cannot be continuous due to the jump in the filtration of $X$ at zero, cf. Theorem \ref{thm:rigv} below.   
\end{enumerate}
\end{example}

We note that we have the following isometry for all step-functions $f$ and $g$:
$$
\E\left[\int_0^T f(t)\d X_t\, \int_0^T g(t)\d X_t\right]
=
\int_0^T \K^*[f](t)\K^*[g](t)\, \d m(t).
$$
By using this isometry, we can extend the Wiener-integral with respect to $X$ to the closure of step-functions under this isometry.

Denote
\begin{equation*}\label{eq:Kinv}
K^{-1}(t,s) = (\K^*)^{-1}\left[\1_{[0,t)}\right](s).
\end{equation*}

\begin{theorem}[Invertible Volterra Representation]\label{thm:rigv}
Let $X$ be a continuous regular invertible Gaussian Volterra process. Then the process
\begin{equation}\label{eq:M_from_X}
M_t = \int_0^t K^{-1}(t,s)\, \d X_s, \quad t\in[0,T],
\end{equation}
is a continuous Gaussian martingale with bracket $m$, and $X$ can be recovered from it by
\begin{equation}\label{eq:X_from_M}
X_t = \int_0^t K(t,s)\, \d M_s, \quad t\in[0,T].
\end{equation}
\end{theorem}

The martingale $M$ in Theorem \ref{thm:rigv} is called a \emph{fundamental martingale}.  Clearly, it is not unique.  

\begin{proof}
Only the continuity of $M$ is unclear.  However, the continuity of $M$ is equivalent of the continuity of $m$. Indeed, the Gaussian martingale $M$ can be realized as $M_t = W_{m(t)}$, where $W$ is a Brownian motion.  
\end{proof}

\begin{remark}[Continuity]
We remark that we assumed the continuity of $X$ \emph{a priori}. In general, the process $X$ is always $L^2$-continuous. Indeed, this follows directly from the It\^o isometry. However, this does not necessarily imply almost surely continuous sample paths, as the modulus of continuity in $L^2$ depends on the function $m$ which in general may behave badly. On the other hand, if $m$ is absolutely continuous with respect to the Lebesgue measure, then $X$ is even H\"older continuous.
\end{remark}

\section{Quadratic Variation, Arbitrage and Completeness}\label{sect:arbitrage-completeness-qv}

The form of the solution of risky-asset dynamics \eqref{eq:risky-dynamics} depend on the quadratic variation of the noise process $X$.  Recall that the (pathwise) quadratic variation of a process $X$ is defined as 
$$
q^2(t) = {\langle X \rangle}_t = 
\lim_{n\to\infty} \sum_{t_{k}^n\le t} \left(X_{t_k^n} - X_{t_{k-1}^n}\right)^2,
$$
where $\{t_0^n=0< t_1^n<\cdots < t_n^n = T\}$ is a sequence of partitions of $[0,T]$
such that $\max_k |t_k^n-t_{k-1}^n| \to 0$.

\begin{lemma}[Quadratic Variation]\label{lmm:qv}
For a regular invertible Gaussian Volterra process the quadratic variation always exists.  Furthermore, it is deterministic and given by
\begin{equation*}\label{eq:qv}
q^2(t) = \int_0^t K(s,s)^2\,\d m(s)
\end{equation*}
\end{lemma}

\begin{proof}
By \cite[Theorem 3.1]{Viitasaari-2015-preprint}, the convergence of quadratic variation of a Gaussian process $X$ holds also in $L^p$ for any $p\geq 1$. Suppose first that the quadratic variation is deterministic. Then, by using representation \eqref{eq:X_from_M} we obtain that 
\begin{eqnarray*}
\lefteqn{\E\left[(X_t - X_{t-\Delta t})^2\right]} \\ 
&=& 
\int_0^T \left(K(t,u) - K(t-\Delta t, u)\right)^2\, \d m(u) \\
&=&
\int_{t-\Delta t}^t K(t,u)^2\, \d m(u) 
+ 
\int_0^{t-\Delta t} \left(K(t,u) - K(t-\Delta t, u)\right)^2\, \d m(u).
\end{eqnarray*}
For deterministic quadratic variations the claim follows from this by using Taylor's approximation for the kernel in the latter integral.

It remains to prove that the quadratic variation is deterministic. By \cite[Theorem 3.1]{Viitasaari-2015-preprint}, it suffices to prove that 
\begin{equation}
\label{eq:conv_needed}
\max_{1\leq j\leq n}\sum_{t_{k}^n\le t} \left(X_{t_k^n} - X_{t_{k-1}^n}\right)\left(X_{t_j^n} - X_{t_{j-1}^n}\right) \rightarrow 0.
\end{equation}
Let $k>j$. Representation \eqref{eq:X_from_M} together with the It\^o isometry yields
\begin{eqnarray*}
\lefteqn{\E\left[\left(X_{t_k^n} - X_{t_{k-1}^n}\right)\left(X_{t_j^n} - X_{t_{j-1}^n}\right)\right]}\\
&=&\int_{t_{j-1}^n}^{t^n_j}\left(K(t_k^n,u)-K(t_{k-1}^n,u)\right)K(t_j^n,u)\,\d m(u) \\
&+&\int_0^{t_{j-1}^n}\left(K(t_k^n,u)-K(t_{k-1}^n,u)\right)\left(K(t_j^n,u)-K(t_{j-1}^n,u)\right)\d m(u).
\end{eqnarray*}
For the first term we use the fact that $t\mapsto K(t,u)$ is increasing together with the bound $K(t_j^n,u)\leq K(T,u)$. Hence we observe that summing with respect to either of the variables and letting $\max_k |t_k^n-t_{k-1}^n| \to 0$ yields convergence towards zero. For the second term, it suffices to observe
\begin{eqnarray*}
\lefteqn{\int_0^{t_{j-1}^n}\left(K(t_k^n,u)-K(t_{k-1}^n,u)\right)\left(K(t_j^n,u)-K(t_{j-1}^n,u)\right)\d m(u) }\\
&\leq &\int_0^{T}\left(K(t_k^n,u)-K(t_{k-1}^n,u)\right)\left(K(t_j^n,u)-K(t_{j-1}^n,u)\right)\d m(u).
\end{eqnarray*}
Hence summing with respect to either of the variables and letting $\max_k |t_k^n-t_{k-1}^n| \to 0$ we get \eqref{eq:conv_needed}.
\end{proof}

By Lemma \ref{lmm:qv} and F\"ollmer \cite{Follmer-1981}, the solution to the stochastic differential equation \eqref{eq:risky-dynamics} defining the discounted risky asset price is given by
\begin{equation*}\label{eq:risky-dynamics-solution}
S_t = S_0\exp\left\{\mu(t)-\frac12 q^2(t) + X_t\right\}
\end{equation*}
and the quadratic variation of $S$ is
\begin{equation*}
\langle S \rangle_t = \int_0^t S_s^2 \, \d q^2(s). 
\end{equation*}

Denote
$$
q^2(s,t) = q^2(t)-q^2(s).
$$

Suppose $q^2$ is non-vanishing on every interval.  Then it follows from the robust replication theorem of \cite{Bender-Sottinen-Valkeila-2008} that the pricing model is free of arbitrage under so-called allowed strategies and replications of vanilla claims are robust in the sense that, as replicating strategies are involved, one can replace $X$ with a Gaussian martingale with bracket $q^2$. 
Thus, we have the following proposition:

\begin{proposition}[Robust Hedging]\label{pro:rh}
Let $f(S_T)$ be a European claim.  Then its Markovian replicating strategy is given by the delta-hedge
$$
\pi_t = \frac{\p v}{\p x}(t,S_t),
$$
where
$$
v(t,S_t) = 
\int_{-\infty}^\infty 
f\left(S_t\e^{-\frac12 q^2(t,T)+q(t,T)z}\right)
\phi(z)\d z
$$
is the value of the replicating strategy $\pi$ at time $t$.
\end{proposition}

\begin{proof}
By the robust replication theorem 5.4 of \cite{Bender-Sottinen-Valkeila-2008}
$$
V^\pi_t = \E_\Q\left[ f(S_T)\,\big|\, \F_t\right],
$$
where, under $\Q$, the price process $S$ is the exponential martingale driven by a Gaussian martingale $G$ with bracket $q^2$.  By equality of filtrations, we have
\begin{eqnarray*}
V^\pi_t &=& \E\left[f\left(S_0\e^{G_T-\frac12 q^2(T)}\right)\,\Big|\, \F_t^G\right] \\
&=&
\E\left[f\left(S_t\e^{-\frac12 q^2(t,T)+(G_T-G_t)}\right)\,\Big|\, \F_t^G\right].
\end{eqnarray*} 
The claim follows from this, since $G_T-G_t$ is independent of $\F_t^G$ and is Gaussian with zero mean and variance $q^2(t,T)$.
\end{proof}

\begin{remark}[Black--Scholes Type BPDE]
If $q^2$ is absolutely continuous with respect to the Lebesgue measure then the European vanilla option $f(S_T)$ can be replicated by solving its time-value from the Black--Scholes type backward partial differential equation 
\begin{eqnarray*}
\frac{\partial v}{\partial t}(t,x) + \frac12 x^2 \frac{\d q^2}{\d t}(t)\frac{\partial^2 v}{\partial x^2}(t,x) &=& 0, \\
v(T,x) &=& f(x).
\end{eqnarray*}
\end{remark}

\begin{remark}[Vanishing Quadratic Variation]
Proposition \ref{pro:rh} remains formally true for $q^2\equiv 0$.  However, in this case the replicating strategy is very simple:
$$
\pi_t = f'(S_t).
$$
\end{remark}

\begin{remark}[Simple Arbitrage]
If the quadratic variation measure $q^2$ vanishes on some interval, then there are simple arbitrage opportunities.  Indeed, suppose $q^2(s,t)=0$ for some $0\le s < t \le T$.  Then (cf. \cite{Bender-Sottinen-Valkeila-2007} and \cite{Bender-Sottinen-Valkeila-2011})
$$
(S_t-S_s)^+ = \int_s^t \1_{[S_s,\infty)}(S_u)\, \d S_u.
$$
So, a buy-and-hold-when-expensive strategy would generate arbitrage.
\end{remark}

Finally, let us give a condition for the completeness and freedom of arbitrage for regular invertible Gaussian Volterra noise pricing models. Below in Proposition \ref{pro:cna} the representation in law means that the fundamental martingale $M$ in \eqref{eq:rep-X} does not necessarily have to be the same as in the representation \eqref{eq:X_from_M}, it just have to have the same law.  This is a subtle difference that has little practical consequences.

\begin{proposition}[Completeness and No-Arbitrage]\label{pro:cna}
The pricing model \eqref{eq:risky-dynamics} is complete and free of arbitrage if and only if the quadratic variation $q^2$ given by \eqref{eq:qv} is strictly increasing and $X$ admits the representation in law
\begin{equation}\label{eq:rep-X}
X_t = 
\int_0^t K(s,s)\, \d M_s 
- \int_0^t\int_0^s h(s,u)K(u,u)\d M_u\, K(s,s)^2 \, \d m(s)
\end{equation}
for some Volterra kernel $h$ in $L^2([0,T], \d q^2\times \d q^2)$.  A sufficient condition for this is that the kernel $K$ satisfies
\begin{equation}\label{eq:K-equiv}
K(t,s) =
K(s,s)\left[1 - \int_s^t h(u,s)K(u,u)^2\,\d m(u)\right]
\end{equation}
for some Volterra kernel $h$ in $L^2([0,T], \d q^2\times \d q^2)$.
\end{proposition} 

\begin{proof}
By the fundamental theorems of asset pricing we have to show that there exists a unique (in law) Gaussian martingale $G$ such that $X$ is equivalent to it. Since $X$ has quadratic variation $q^2$ given by \eqref{eq:qv}, the Gaussian martingale must have the same quadratic variation.  Actually, we may assume that 
\begin{equation}\label{eq:G_from_M}
G_t = \int_0^t K(s,s)\, \d M_s.
\end{equation}

Now we can apply the Hitsuda representation theorem in a similar manner as in \cite{Sottinen-2004} or \cite{Sottinen-Tudor-2006}. 
Recall that by Hitsuda representation theorem (see \cite{Hitsuda-1968}) a Gaussian process $\tilde W$ is equivalent to Brownian motion $W$ if and only if it admits the representation (in law)
$$
\tilde W_t = W_t - \int_0^t \int_0^s \ell(s,u)\, \d W_u \, \d s,
$$
where $\ell$ is any Volterra kernel in $L^2([0,T]^2, \d t\times\d t)$.  Since the Gaussian martingale $G$ with bracket $q^2$ is a time-changed Brownian motion, $G_t = W_{q^2(t)}$, we obtain the representation (in law)
\begin{equation}\label{eq:hitsuda-X}
X_t = G_t - \int_0^t \int_0^s h(s,u)\, \d G_u \, \d q^2(s),
\end{equation}
where
$$
h(s,u) = \ell\left(q^2(s), q^2(u)\right). 
$$
Consequently, $X$ is equivalent to $G$ if and only if it admits (in law) representation \eqref{eq:hitsuda-X} with some Volterra kernel $h$ in $L^2([0,T], \d q^2\times \d q^2)$. The representation \eqref{eq:rep-X} follows by combining \eqref{eq:hitsuda-X} with \eqref{eq:G_from_M}.

Equation \eqref{eq:K-equiv} follows by combining \eqref{eq:rep-X} with \eqref{eq:X_from_M}.
\end{proof}

\section{Prediction}\label{sect:prediction}

Prediction of the asset price or the noise is possible because
\begin{enumerate}
\item
all the filtrations $\F_t^S$, $\F_t^X$ and $\F_t^M$ are the same,
\item
for regular invertible Gaussian Volterra processes we can use the theorem of Gaussian correlations in an explicit manner.
\end{enumerate}

Denote
\begin{eqnarray*}
\hat X_t(u) &=& \E\left[ X_t \,\big|\, \F_u^X\right], \\
\hat R(t,s|u) &=& \Cov\left[X_t,X_s\, \big|\, \F_u^X\right].
\end{eqnarray*}

\begin{theorem}[Prediction]\label{thm:prediction}
Let $X$ be a regular invertible Gaussian Volterra process with fundamental martingale $M$. Then the conditional process 
$X_t(u) = X_t |\F_u^X$, $t\in[u,T],$
is Gaussian with $\F_u^X$-measurable mean
\begin{equation*}\label{eq:cond-mean}
\hat X_t(u) =
X_u - \int_0^u \Psi(t,s|u)\, \d X_s,
\end{equation*}
where
$$
\Psi(t,s|u) = (\K^*)^{-1}\left[K(t,\cdot)-K(u,\cdot)\right](s),
$$
and deterministic covariance
\begin{equation*}\label{eq:conv-cov}
\hat R(t,s|u) =
R(t,s) - \int_0^u K(t,v)K(s,v)\,  m(\d v).
\end{equation*}
\end{theorem}

\begin{proof}
Consider first the conditional mean.  By Theorem \ref{thm:rigv}
\begin{eqnarray*}
\hat X_t(u) &=& \E\left[\int_0^t K(t,s)\, \d M_s \,\Big|\, \F_u^M\right] \\
&=&
\int_0^u K(t,s)\, \d M_s \\
&=&
\int_0^t K(t,s)\, \d M_s - \int_0^u \left[K(t,s)-K(u,s)\right]\, \d M_s \\
&=&
X_t - \int_0^u  \left[K(t,s)-K(u,s)\right]\, \d M_s.
\end{eqnarray*}
The conditional expectation formula follows from this by using the isometric definition of Wiener integration with respect to $X$.

Consider then the conditional variance. By Theorem \ref{thm:rigv} and calculations above
\begin{eqnarray*}
\hat R(t,s|u) &=&
\E\left[\int_u^t K(t,v)\, \d M_v \int_u^s K(s,v)\, \d M_v \,\Big|\, \F_u^M \right] \\
&=&
\E\left[\int_u^t K(t,v)\, \d M_v \int_u^s K(s,v)\, \d M_v \right] \\
&=&
\int_u^{t\wedge s} K(t,v) K(s,v)\, m(\d v) \\
&=&
\int_0^{t \wedge s} K(t,v)K(s,v)\, m(\d v) - \int_0^u K(t,v)K(s,v)\, m(\d v) \\
&=&
R(t,s) - \int_0^u K(t,v)K(s,v)\, m(\d v).
\end{eqnarray*}
\end{proof}

Denote
\begin{eqnarray*}
\hat\rho(t|u) &=& \sqrt{\hat R(t,t|u)}, \\
\beta(u,t)
&=&
\mu(u,t) - \frac12 q^2(u,t)
\end{eqnarray*}
Then
$$
S_t = S_u\e^{\beta(u,t) + (X_t-X_u)}
$$
and
$$
\Var[X_t-X_u\,\big|\, \F_u] = \hat\rho^2(t|u).
$$

The following corollary is the key result that allows us to calculate the conditional-mean hedging strategies in Section \ref{sect:cmh}.

\begin{corollary}[Prediction]\label{cor:prediction}
Let $0\le u \le t\le T$.
Let $f\colon[0,T]\times\R\to\R$ be such that $f(t,S_t)$ is integrable.
Let $\phi$ be the standard Gaussian density function. Then
$$
\E\left[f(t,S_t)\,\Big|\, \F_u^S\right]
=
\int_{-\infty}^\infty 
f\left(t,S_u\e^{\beta(u,t)-\int_0^u \Psi(t,s|u)\d X_s+ \hat \rho(t|u)z}\right)\phi(z)\d z.
$$
\end{corollary}

\begin{proof}
Given Theorem \ref{thm:prediction}, the equality of filtrations and the F\"ollmer--It\^o formula, the claim follows from straightforward calculations:
\begin{eqnarray*}
\E\left[f(t,S_t)\,\Big|\, \F_u^S\right]
&=&
\E\left[f\left(t,S_u\e^{\beta(u,t) + \left(X_t-X_u\right)}\right)\,\Big|\, \F_u^X\right] \\
&=&
\int_{-\infty}^\infty 
f\left(t,S_u\e^{\beta(u,t) + \left(\hat\rho(t|u)z+\hat X_t(u)\right)-X_u} \right)\phi(z)\d z \\
&=&
\int_{-\infty}^\infty 
f\left(t,S_u\e^{\beta(u,t) + \hat\rho(t|u)z+\left(\hat X_t(u)-X_u\right)} \right)\phi(z)\d z \\
&=&
\int_{-\infty}^\infty 
f\left(t,S_u\e^{\beta(u,t) + \hat\rho(t|u)z-\int_0^u \Psi(t,s|u)\d X_s} \right)\phi(z)\d z,
\end{eqnarray*}
proving the claim.
\end{proof}

\section{Conditional-Mean Hedging}\label{sect:cmh}

We are interested in the pricing and hedging of European vanilla options $f(S_T)$ of the single discounted underlying asset $S=(S_t)_{t\in[0,T]}$, where $T>0$ is a fixed time of maturity of the option.

We assume that the trading only takes place at fixed preset time points $0=t_0<t_1<\cdots<t_N<T$. We denote by $\pi^N$ the discrete trading strategy
$$
\pi^N_t = \pi^N_0\1_{\{0\}}(t) + \sum_{i=1}^N \pi_{t_{i-1}}^N\1_{(t_{i-1},t_{i}]}(t).
$$
The value of the strategy $\pi^N$ is given by
\begin{equation}\label{eq:kappa}
V^{\pi^N,k}_t = V^{\pi^N,k}_0 + \int_0^t \pi^N_u \, \d S_u -
\int_0^t kS_{u} |\d\pi^N_{u}|,
\end{equation}
where $k\in[0,1)$ is the proportional transaction cost.

Under transaction costs perfect hedging is not possible.  In this case, it is natural to try to hedge on average in the sense of the following definition:

\begin{definition}[Conditional-Mean Hedge]\label{dfn:cm-hedge}
Let $f(S_T)$ be a European vanilla type option with convex or concave payoff function $f$.  Let $\pi$ be its Markovian replicating strategy: $\pi_t=g(t,S_t)$.  We call the discrete-time strategy $\pi^N$ a \emph{conditional-mean hedge}, if for all trading times $t_i$,
\begin{equation}\label{eq:cm-hedge}
\E\left[ V^{\pi^N, k}_{t_{i+1}}\,|\, \F_{t_i}\right] = \E\left[ V^{\pi}_{t_{i+1}}\,|\, \F_{t_i}\right].
\end{equation}
Here $\F_{t_i}$ is the information generated by the asset price process $S$ up to time $t_i$. 
\end{definition}

\begin{remark}[Conditional-Mean Hedge as Tracking Condition]
Criterion \eqref{eq:cm-hedge} is actually a tracking requirement.  We do not only require that the conditional means agree on the last trading time before the maturity, but also on all trading times.  In this sense the criterion has an ``American'' flavor in it.  From a purely ``European'' hedging point of view, one can simply remove all but the first and the last trading times.
\end{remark}

\begin{remark}[Arbitrage and Uniqueness of Conditional-Mean Hedge]
Note that the  conditional-mean hedging strategy $\pi^N$ depends on the continuous-time hedging strategy $\pi$. 
Since there may be strong arbitrage in the pricing model (zero can be perfectly replicated with negative initial wealth), the replicating strategy $\pi$ may not be unique. However, the strong arbitrage strategies are very complicated.  Indeed, it follows directly from the F\"ollmer--It\^o change-of-variables formula that in the class of Markovian strategies $\pi_t = g(t,S_t)$, the delta-hedge coming from the Black--Scholes type backward partial differential equation is the unique replicating strategy for the claim $f(S_T)$.  
\end{remark}

\begin{remark}[No Martingale Measures]
We stress that the expectation in \eqref{eq:cm-hedge} is with respect to the true probability measure; not under any equivalent martingale measure.  Indeed, equivalent martingale measures may not even exist.
\end{remark}

To find the solution to \eqref{eq:cm-hedge} one must be able to calculate the conditional expectations involved. 

Let $\pi$ be the continuous-time Markovian hedging strategy of the claim $f(S_T)$ and let $V^\pi$ be its value process. Denote
\begin{eqnarray*}
\Delta\hat X_{t_{i+1}}(t_i) 
&=& \hat X_{t_{i+1}}(t_i) - X_{t_i}, \\
&=& \E\left[ X_{t_{i+1}} |\F_{t_i}\right] -  X_{t_i}, \\ 
\Delta\hat S_{t_{i+1}}(t_i) 
&=&
\hat S_{t_{i+1}}(t_i) - S_{t_i} \\
&=&
\E\left[S_{t_{i+1}}|\F_{t_i}\right] - S_{t_i}, \\
\Delta\hat V^\pi_{t_{i+1}}(t_i) 
&=&
\hat V^\pi_{t_{i+1}}(t_i) - V^\pi_{t_i} \\
&=&
\E\left[V^\pi_{t_{i+1}}|\F_{t_i}\right] - V^\pi_{t_i}, \\
\Delta\hat V^{\pi^N,k}_{t_{i+1}}(t_i) 
&=&
\hat V^{\pi^N,k}_{t_{i+1}}(t_i) - V^\pi_{t_i} \\
&=&
\E\left[V^{\pi^N,k}_{t_{i+1}}|\F_{t_i}\right] - V^{\pi^N,k}_{t_i}. 
\end{eqnarray*}

Denote
$$
\gamma(s,t,T) = \beta(s,t) - \frac12 q^2(t,T).
$$

Lemma \ref{lmm:c-deltas} below states that all these conditional gains listed above can be calculated explicitly.  

\begin{lemma}[Conditional Gains]\label{lmm:c-deltas}
\begin{eqnarray*}
\Delta\hat X_{t_{i+1}}(t_i)
&=&
-\int_0^u \Psi(t,s|u)\, \d X_u, \\
\Delta\hat S_{t_{i+1}}(t_i) 
&=&
S_{t_i}\left(\e^{\beta(t_i,t_{i+1})+\frac12\hat\rho^2(t_{i+1}|t_i) + \Delta\hat X_{t_{i+1}}(t_i)}
-1\right), \\
\Delta\hat V^\pi_{t_{i+1}}(t_i)
&=&
\int_{-\infty}^\infty 
\left[\int_{-\infty}^\infty 
f\left(S_{t_{i}}\e^{\gamma(t_i,t_{i+1},T) + \hat\rho(t_{i+1}|t_i)y +q(t_{i+1},T)z}\right)
\, \phi(y)\d y \right.\\
& & - \left.
f\left(S_{t_i}\e^{-\frac12 q^2(t_i,T) + q(t_i,T)z}\right)\right] 
\phi(z)\d z, \\
\Delta\hat V^{\pi^N,k}_{t_{i+1}}(t_i)
&=&
\pi_{t_i}^N\Delta\hat S_{t_{i+1}}(t_i) - k S_{t_i}|\Delta \pi^N_{t_i}|.
\end{eqnarray*}
\end{lemma}

\begin{proof}
The formula for $\Delta\hat X_{t_{i+1}}(t_i)$ is given by Theorem \ref{thm:prediction}.

Consider $\Delta \hat S_{t_{i+1}}(t_i)$. By Corollary \ref{cor:prediction},
\begin{eqnarray*}
\hat S_{t_{i+1}}(t_i)
&=&
\int_{-\infty}^\infty S_{t_i} \e^{\beta(t_i,t_{i+1}) + \Delta\hat X_{t_{i+1}}(t_i)+ \hat\rho(t_{i+1}|t_i)z}\, \phi(z)\d z \\
&=&
S_{t_i}\e^{\beta(t_i,t_{i+1})+ \Delta\hat X_{t_{i+1}}(t_i)}\int_{-\infty}^\infty \e^{\hat\rho(t_{i+1}|t_i)z}\, \phi(z)\d z \\
&=&
S_{t_i}\e^{\beta(t_i,t_{i+1})+\frac12\hat\rho^2(t_{i+1}|t_i) + \Delta\hat X_{t_{i+1}}(t_i)}.
\end{eqnarray*}
Consequently,
$$
\Delta\hat S_{t_{i+1}}(t_i) 
=
S_{t_i}\left(\e^{\beta(t_i,t_{i+1})+\frac12\hat\rho^2(t_{i+1}|t_i) + \Delta\hat X_{t_{i+1}}(t_i)}
-1\right).
$$

Consider then $\Delta\hat V^\pi_{t_{i+1}}(t_i)$. By Proposition \ref{pro:rh} and Fubini theorem,
$$
\hat V_{t_{i+1}}(t_i)
=
\int_{-\infty}^\infty 
\E\left[f\left(S_{t_{i+1}}\e^{-\frac12 q^2(t_{i+1},T)+q(t_{i+1},T)z}\right)
\,\Big|\, \F_{t_i}\right]\phi(z)\d z.
$$
Now,
\begin{eqnarray*}
\lefteqn{\E\left[f\left(S_{t_{i+1}}\e^{-\frac12 q^2(t_{i+1},T)+q(t_{i+1},T)z}\right)
\,\Big|\, \F_{t_i}\right] } \\
&=& 
\E\left[f\left(S_{t_{i}}\e^{\beta(t_i,t_{i+1})+\frac12 q^2(t_i,T) + (X_{t_{i+1}}-X_{t_i})+q(t_{i+1},T)z}\right)
\,\Big|\, \F_{t_i}\right] \\
&=&
\int_{-\infty}^\infty 
f\left(S_{t_{i}}\e^{\beta(t_i,t_{i+1}) - \frac12 q^2(t_{i+1},T) + \hat\rho(t_{i+1}|t_i)y +q(t_{i+1},T)z}\right)
\, \phi(y)\d y \\
&=&
\int_{-\infty}^\infty 
f\left(S_{t_{i}}\e^{\beta(t_i,t_{i+1}) - \frac12 q^2(t_{i+1},T)  + \hat\rho(t_{i+1}|t_i)y +q(t_{i+1},T)z}\right)
\, \phi(y)\d y.
\end{eqnarray*}
Since
$$
V^\pi_{t_i} = \int_{-\infty}^\infty f\left(S_{t_i}\e^{-\frac12 q^2(t_i,T) + q(t_i,T)z}\right) \phi(z)\d z,
$$
we obtain 
\begin{eqnarray*}
\lefteqn{\Delta\hat V^\pi_{t_{i+1}}(t_i)} \\
&=&
\int_{-\infty}^\infty \int_{-\infty}^\infty 
f\left(S_{t_{i}}\e^{\gamma(t_i,t_{i+1},T)  + \hat\rho(t_{i+1}|t_i)y +q(t_{i+1},T)z}\right)
\, \phi(y)\d y\phi(z)\d z \\
& &
-
\int_{-\infty}^\infty f\left(S_{t_i}\e^{-\frac12 q^2(t_i,T) + q(t_i,T)z}\right) \phi(z)\d z \\
&=& 
\int_{-\infty}^\infty 
\left[\int_{-\infty}^\infty 
f\left(S_{t_{i}}\e^{\gamma(t_i,t_{i+1},T) + \hat\rho(t_{i+1}|t_i)y +q(t_{i+1},T)z}\right)
\, \phi(y)\d y \right.\\
& & - \left.
f\left(S_{t_i}\e^{-\frac12 q^2(t_i,T) + q(t_i,T)z}\right)\right] 
\phi(z)\, \d z.
\end{eqnarray*}

Finally, we calculate
\begin{eqnarray*}
\hat V^{\pi^N,k}_{t_{i+1}}(t_i)
&=&
\E\left[ V^{\pi^N,k}_{t_{i+1}}\, \big| \, \F_{t_i}\right] \\
&=&
V^{\pi^N,k}_{t_i} + 
\E\left[\int_{t_{i}}^{t_{i+1}} \pi^N_u \, \d S_u - \int_{t_i}^{t_{i+1}} k S_{u}|\d\pi^N_{u}|\,\Big|\, \F_{t_i}\right] \\
&=&
V^{\pi^N,k}_{t_i} + \pi_{t_i}^N\left(\E\left[S_{t_{i+1}}\big| \F_{t_i}\right] - S_{t_i}\right) - k S_{t_i}|\Delta \pi^N_{t_i}| \\
&=&
V^{\pi^N,k}_{t_i} + \pi_{t_i}^N\Delta\hat S_{t_{i+1}}(t_i) - k S_{t_i}|\Delta \pi^N_{t_i}|.
\end{eqnarray*}
The formula for $\Delta\hat V^{\pi^N,k}_{t_{i+1}}(t_i)$ follows from this.
\end{proof}

Now we are ready to state and prove our main result.  We note that, in principle, our result is general: it is true in any pricing model where the option $f(S_T)$ can be replicated. In practice, our result is specific to the regular invertible Gaussian Volterra noise pricing model via Lemma \ref{lmm:c-deltas}.

\begin{theorem}[Conditional-Mean Hedging Strategy]\label{thm:dhedging}
The conditional mean hedge of the European vanilla type option with convex or concave positive payoff function $f$ with proportional transaction costs $k$ is given by the recursive equation
\begin{equation}\label{eq:dhedging}
\pi^N_{t_i}
=
\frac{\Delta\hat V^\pi_{t_{i+1}}(t_i) + (V^\pi_{t_i}-V^{\pi^N,k}_{t_i}) + k S_{t_i}|\Delta\pi^N_{t_i}|}{\Delta\hat S_{t_{i+1}}(t_i)},
\end{equation}
where $V^{\pi^N,k}_{t_i}$ is determined by \eqref{eq:kappa}.
\end{theorem}

\begin{proof}
Let us first consider the left hand side of \eqref{eq:cm-hedge}. We have
\begin{eqnarray*}
\E\left[V^{\pi^N,k}_{t_{i+1}}\,\big|\, \F_{t_i}\right] 
&=&
\E\left[V^{\pi^N,k}_{t_{i}} + \int_{t_i}^{t_{i+1}}\pi^N_u\, \d S_u
-k \int_{t_i}^{t_{i+1}} S_u |\d\pi_u^N|\,\Big|\, \F_{t_i}\right] \\ 
&=&
V^{\pi^N,k}_{t_i}
+ \pi^N_{t_i} \E\left[ S_{t_{i+1}}(t_i) - S_{t_i}\,\big|\, \F_{t_i}\right] - kS_{t_i}|\Delta\pi^N{t_i}| \\
&=&
V^{\pi^N,k}_{t_i}
+ \pi^N_{t_i} \Delta\hat S_{t_{i+1}}(t_i) - kS_{t_i}|\Delta\pi^N_{t_i}|.
\end{eqnarray*}
For the right-hand-side of \eqref{eq:cm-hedge}, we simply write 
\begin{eqnarray*}
\E\left[V^{\pi}_{t_{i+1}}\,\big|\, \F_{t_i}\right]
&=&
\Delta\hat V^\pi_{t_{i+1}}(t_i) + V^\pi_{t_i}.
\end{eqnarray*}
Equating the sides we obtain \eqref{eq:dhedging} after a little bit of simple algebra.
\end{proof}

\begin{remark}[Interpretation]
Taking the expected gains $\Delta\hat S_{t_{i+1}}(t_i)$ to be the num\'eraire, one recognizes three parts in the hedging formula \eqref{eq:dhedging}.   First, one invests on the expected gains in the time-value of the option. This ``conditional-mean delta-hedging'' is intuitively the most obvious part.  Indeed, a na\"ive approach to conditional-mean hedging would only give this part, forgetting to correct for the tracking-errors already made, which is the second part in \eqref{eq:dhedging}.  The third part in \eqref{eq:dhedging} is obviously due to the transaction costs. 
\end{remark}

\begin{remark}[Initial Position]
Note that the equation \eqref{eq:dhedging} for the strategy of the conditional-mean hedging is recursive: in addition to the filtration $\F_{t_i}$, the position $\pi^N_{t_{i-1}}$ is needed to determine the position $\pi^N_{t_i}$. Consequently, to determine the conditional-meand hedging strategy by using \eqref{eq:dhedging}, the initial position $\pi^N_0$ must be fixed.  The initial position is, however, not uniquely defined.  Indeed, let $\beta^N_0$ be the position in the riskless asset. Then the conditional-mean criterion \eqref{eq:cm-hedge} only requires that
$$
\beta^N_0 + \pi^N_0\E[S_{t_1}] -kS_0|\pi^N_0|
=
\E[V^\pi_{t_1}].
$$
There are of course infinite number of pairs $(\beta^N_0,\pi^N_0)$ solving this equation.  A natural way to fix the initial position $(\beta^N_0,\pi^N_0)$ for the investor interested in conditional-mean hedging would be the one with minimal cost.  If short-selling is allowed, the investor is then faced with the minimization problem
$$
\min_{\pi^N_0 \in \R} w(\pi^N_0),
$$
where the initial wealth $w$ is the piecewise linear function
\begin{eqnarray*}
w(\pi^N_0)
&=&
\beta^N_0 + \pi^N_0 S_0 \\
&=&
\left\{\begin{array}{rl}
\E[V^\pi_{t_1}] - \left(\Delta \hat S_{t_1}(0)-kS_0\right)\pi^N_0, 
& \mbox{ if }\quad \pi^N_0\ge 0,\\
\E[V^\pi_{t_1}] - \left(\Delta \hat S_{t_1}(0)+kS_0\right)\pi^N_0,
& \mbox{ if }\quad \pi^N_0 < 0.
\end{array}\right. 
\end{eqnarray*}
Clearly, the minimal solution $\pi^N_0$ is independent of $\E[V^\pi_{t_1}]$, and, consequently, of the option to be replicated.   Also, the minimization problem is bounded if and only if 
\begin{eqnarray*}
k \ge \left|\frac{\Delta\hat S_{t_1}(0)}{S_0}\right|,
\end{eqnarray*}
i.e. the proportional transaction costs are bigger than the expected return on $[0,t_1]$ of the stock.  In this case, the minimal cost conditional mean-hedging strategy starts by putting all the wealth in the riskless asset.
\end{remark}

\bibliographystyle{siam}
\bibliography{pipliateekki}
\end{document}